%% file: main_arxiv.tex
\documentclass[12pt]{article}   	
\usepackage[T1]{fontenc}
\usepackage[latin9]{inputenc}
\usepackage{geometry}                		
\usepackage{graphicx}				
								
\usepackage{amsthm}
\usepackage{amsmath}
\usepackage{authblk,array,setspace} 
\usepackage{amssymb,natbib}
\usepackage[pdftex,colorlinks=true,linkcolor=blue,citecolor=blue,urlcolor=blue,bookmarks=false,pdfpagemode=None]{hyperref}

\newtheorem{theorem}{Theorem}[section]

\newtheorem{proposition}[theorem]{Proposition}

\newtheorem{definition}[theorem]{Definition}

\newcommand{\E}{{\rm E}}
\newcommand{\median}{{\rm median}}
\newcommand{\mode}{{\rm mode}}
\newcommand{\obs}{{\rm obs}}
\newcommand{\mis}{{\rm mis}}

\begin{document} 
\title{Causal inference for ordinal outcomes\protect\thanks{Alexander M. Volfovsky is an NSF Postdoctoral Fellow in the Department of Statistics at Harvard University (volfovsky@fas.harvard.edu). Edoardo M.~Airoldi is an Associate Professor of Statistics at Harvard (airoldi@fas.harvard.edu). Donald B.~Rubin is the John L. Loeb Professor of Statistics at Harvard (rubin@stat.harvard.edu).}} 
\author{Alexander Volfovsky}
\author{Edoardo M. Airoldi}
\author{Donald B. Rubin}
\affil{Department of Statistics, Harvard University}


\maketitle

\newpage
\begin{abstract}
Many outcomes of interest in the social and health sciences, as well as in modern applications in computational social science and experimentation on social media platforms, are ordinal and do not have a meaningful scale. Causal analyses that leverage this type of data, termed ordinal non-numeric, require careful treatment, as much of the classical potential outcomes literature is concerned with estimation and hypothesis testing for outcomes whose relative magnitudes are well defined. Here, we propose a class of finite population causal estimands that depend on conditional distributions of the potential outcomes, and provide an interpretable summary of causal effects when no scale is available. We formulate a relaxation of the Fisherian sharp null hypothesis of constant effect that accommodates the scale-free nature of ordinal non-numeric data. We develop a Bayesian procedure to estimate the proposed causal estimands that leverages the rank likelihood. We illustrate these methods  with an application to educational outcomes in the General Social Survey.
\newline
\vfill


\end{abstract}

\newpage
\tableofcontents

\newpage
\input{intro}
\input{estimands}
\input{tests}
\input{estimation}

\input{model}
\input{disc}

\bibliographystyle{apalike}
\bibliography{biblio}
\end{document}

%% file: intro.tex
\section{Introduction}
Outcomes in the social and economic sciences are frequently ordered, however, it is not
always the case that the scale or magnitude of the outcomes is available.
That is, although outcome $Y_i$ might belong to category one, labeled ``low'', and outcome
$Y_{i^\prime}$ belongs to category two, labeled ``high'', the only information about the relative 
relationship of the two outcomes is that $Y_i<Y_{i^\prime}$. 
Data that exhibit such a structure are termed 
{\it ordinal non-numeric}, and although the categories are frequently represented by 
integer values, there is no substantive information in the
data about the relative magnitude of the categories. Examples of
such data abound in education research, e.g., the 
level of education: ``high school diploma'', ``college'',``masters'',
``PhD'' \citep{dubow2009long}, 
in operations research, in employment records, 
e.g., for job seniority data: ``staff'', 
``manager'', ``vice president'', ``president'' \citep{singh1990jobs}, and in healthcare
research, e.g., when the outcome is pain: ``none'', ``mild'', ``severe'' \citep{collins1997visual}. 
In this context, it is frequently of interest to
make causal statements about how some treatment might affect
an individual's category, e.g., whether a drug reduces pain level
or whether a vocational program leads to better job prospects \citep{mealli2012statistical}. 

The first step when attempting any causal analysis is the 
choice of an estimand, or inferential target, which is the object of interest. When outcomes are 
continuous, the most commonly studied estimand is the average treatment
effect \citep[e.g., see][]{rubin1974estimating,holland1986statistics}. However, this quantity is not well defined for non-numeric
data as the notion of an average of more than two ordinal values 
is not well defined. Other measures
of centrality might be of interest; for example, the modal difference
between treatment and control describes the most common number
of categories that are changed due to treatment. 
However, measures
of centrality for  non-numeric data do not provide a complete picture
because the relative magnitude between pairs of categories is not well defined.
For example, a modal treatment effect of zero, indicating that
most often there is no change due to treatment, might conceal 
information that the treatment is effective for the subsets of individuals whose 
potential outcome equals category one, under control, but is ineffective for anyone  else.
These issues 
motivate our development of a class of conditional estimands in 
Section \ref{sec:estimand}.

The causal inference literature based on the potential outcomes 
framework \citep{rubin1974estimating} has focused on a special case
of ordinal non-numeric response: the binary outcome. This is a
special case because although the categories are ordered and 
non-numeric (for example in medicine: $0=$~dead, $1=$~alive), there
are no relative magnitudes to consider \citep{rosenbaum1983assessing},
that is, the average of binary outcomes is simply the proportion. 
The presence of more outcome levels has frequently led to
model based causal analysis \citep{rubin1978bayesian} without first defining estimands
based on the potential outcomes \citep{shaikh2005threshold,cunha2007identification}. 
A tempting advantage of model-based 
inference for ordinal non-numeric outcomes is the availability
of a continuous latent variable formulation of the outcomes.
Within the context of the potential outcomes framework,
a model-based approach assumes the existence of a 
continuous potential outcome $Z_i(t)$ and a mapping
$g:Z_i(t)\rightarrow Y_i(t)$ that discretizes $Z_i(t)$ into an ordinal non-numeric observation, $Y_i(t)$. We refer to the former as potential outcomes on the latent scale, or as latent potential outcomes, to distinguish them from the actual potential outcomes, which are observed on the measurement scale. 
It is possible to define estimands
for either of these types of potential outcomes, 
but we show that the latent variable formulation suffers
from undesirable identifiability issues
in Section \ref{sec:estimand}. 

The remainder of this paper is structured as follows. In 
Section \ref{sec:estimand} we formally revisit the potential
outcomes framework and describe various interesting
estimands. In Section \ref{sec:test} we
revisit the framework for the Fisher exact test for  hypothesis sharp null
of no causal effect. Unlike in the continuous case, the sharp null
of constant (non zero) effect is not available for ordinal non-numeric
data and so we develop a novel test for the related null hypothesis of
an effect that changes at most one category.
Section \ref{sec:estimation} describes estimation
procedures for  estimands on the observed scale with and without additional 
assumption on the joint distribution of the potential outcomes. 
In Section \ref{sec:bayes} we outline two Bayesian procedures 
for estimating causal effects on the observed scale, one within the modeling framework
of a standard ordered probit and one based on the rank 
likelihood. These methods are illustrated through a practical example based on data from the General Social
Survey, in Section~\ref{dataexamp}. Some concluding remarks follow. 


%% file: estimands.tex
\section{Potential outcomes framework and estimands}
\label{sec:estimand}
In this section we provide a general introduction to the
potential outcomes framework for causal inference, frequently
referred to as the Rubin Causal Model \citep{rubin1974estimating,rubin1978bayesian}. 
A detailed history and exposition is available in \cite{rubin2005causal}. 

The formal potential outcomes framework provides a 
clear separation between the science, i.e., the object of inference, and the
process by which inference about the science is made. Brief definitions
of specific concepts are as follows. The unit of inference is a
physical object at a particular point in time, e.g., an individual. For a binary treatment
$T\in \{0,1\}$ and $N$ units, a table that codifies the {science} 
under the Stable-Unit-Treatment-Value Assumption (SUTVA, \citet{rubin1980randomization})
is an
$N\times 2$ table where each row equals $(Y_i(0),Y_i(1))$, where $Y_i(0)$ is the potential
outcome for unit $i$ under control and $Y_i(1)$ is the potential outcome for unit
$i$ under treatment. A unit level causal effect is a comparison of the potential
outcomes for a given unit. However, we
can only observe one of the potential outcomes for each unit---that is,
the treatment status of unit $i$ is either $W_i=0$ or $W_i=1$. 
As such, unit specific causal effects cannot be observed and must
be inferred. This is facilitated by different assumptions
that can be imposed on the marginal and joint distributions of the treatment
assignments $W_i$ and the science table. 

One of the most commonly studied causal estimands, both in terms of 
theory and applications \citep{rubin1974estimating,lin2013agnostic}
 is the average causal effect $\E[\bar Y(1)- \bar Y(0)]$.
Under relatively mild conditions \citep[e.g., see][]{imbens2014causal}
 this quantity can be estimated in both
randomized and observational studies. A fundamental property of 
this estimand is that its dependence on unit level causal effects is 
separable into a dependence on only the potential
outcomes under treatment and only the potential outcomes under control
 due to the linearity of the expectation. As we will see below, many estimands
 that are of interest when dealing with ordinal non-numeric outcomes do not 
 possess this property. We argue that this particular type of average causal effect is not
 meaningful for ordinal non-numeric outcomes. 

For ordinal non-numeric data,
one can 
consider estimands on the observed scale or,
 if assuming 
the existence of a latent variable representation of the 
science table,
on a latent scale. Recall that on the latent scale
we make the assumption that there are latent potential outcomes
$(Z_i(0),Z_i(1))$ that are related to the observed scale via a deterministic function
$g:Z_i(t)\rightarrow Y_i(t),\ t=0,1$. The mapping is dependent on the ordering on the observed scale
in the following sense; if there are categories that
admit different orderings based on different
criteria then any inference must be conditional on the choice
of a specific criterion (e.g. it takes longer to become an MD than it does
to become a lawyer but the starting salary of lawyers 
is higher than that of doctors). 
This suggests that the function $g$ must be conditioned
on the ordering criterion, on the observed scale, to maintain
the appropriate ordering of potential outcomes
on the latent scale. An additional complication
for causal inference using the latent potential outcomes
is that the function $g$ is rarely known explicitly,
making many estimands on the 
latent scale difficult to interpret. 
In what follows
we first discuss estimands of interest on the
observed scale, arguing for the use of conditional
estimands to best capture the effect of treatment in
ordinal non-numeric data. We then describe estimands
on the latent  
scale and note that due to the complications
associated with the latent scale, the
appropriate estimands for most applications
are on the observed scale.

Throughout,
the treatment is binary, there are $k$ outcome levels on the observed scale, 
and the latent scale is continuous and unbounded on the real line.

\subsection{Estimands on the observed scale}\label{rawscale}
As with the general potential outcomes framework, the 
complete information about any causal effect is found in the
joint distribution of the potential outcomes;
here we ignore the possible presence of covariates. 
For the observed scale
this is summarized by a $k\times k$ matrix $P$ where $p_{ij}=\Pr(Y(0)=i,Y(1)=j)$.
All estimands are functions of this matrix $P$. For example,
the pair of marginal distributions, $(P_0, P_1) \equiv (\Pr(Y(0), \Pr(Y(1))) = (P\mathbf{1},\mathbf{1}^tP)$,
is a $2k$ dimensional summary of the joint matrix $P$. If the potential
outcomes are independent, then these marginals encode all the information
in the joint distribution. The high dimensional nature of these estimands
 reduces their simplicity and so we strive
for lower dimensional summaries that are more easily communicated. 

\paragraph{One dimensional estimands.}
The average treatment effect (ATE) mentioned above, a popular scalar estimand, is not meaningful
in the setting of ordinal non-numeric outcomes because the expectation
is not well defined. Other measures of centrality suffer from a similar
degeneracy as they effectively assign a scale to the differences. For example, for a pair of units $i$ and $j$ it might happen that 
$Y_i(1)-Y_i(0)=4-3=Y_j(1)-Y_j(0) =2-1$. Thus both of these differences would contribute the same information
to any function $f$ that only depends on the difference. Similarly
the difference in the medians under treatment and under control,
$\median[Y(1)]-\median[Y(0)]$ obscures the scaling issue.
As such, the only one dimensional summary that is meaningful when the
scale of the observations is not identified, is one that summarizes the difference between the 
marginal distributions
of the potential outcomes.
Formally, let $d(\cdot,\cdot)$ be a metric
on the space of probabilities and define the estimand $d(P_0,P_1)$. This estimand
is especially important when considering the sharp null hypothesis
of no effect, as we show in Section \ref{sec:test}.
%

\paragraph{Multidimensional estimands.}
As one dimensional summaries are often insufficient for providing
adequate information about any causal effects between
treatment and control groups, multidimensional estimands
are required. A two dimensional estimand $(f_0[Y(0)],f_1[Y(1)])$, for instance,
where the functions $f_0,f_1$ are independent of the scale
of the potential outcomes, provides information that might have been concealed
by estimands that considered the difference $f_1(\cdot)-f_0(\cdot)$.
As argued above, while differences provide overarching information about outcomes
under control and treatment, they do not provide any information
on the amount of category change due to treatment. 
Another two-dimensional estimand that provides a compact summary
of this
information is the most likely pair of potential outcomes,
${\rm ml}=\arg\max_{i,j}p_{ij}$. In general, this is a function of $P$,
but in cases of independence between the potential outcomes
it becomes a low dimensional summary of $P_0$ and $P_1$. 

In practice, we are interested in the effects of
mechanisms by which treatment changes the potential outcome
under control. As such, it is natural to consider estimands
that condition on the level of the potential outcome
under control. Here, we advocate for the use
of a class of causal estimands
that involves the conditional probabilities
of the potential outcomes.
That is, we might be interested in the $k$-dimensional summaries
$M_{1i}=\median[Y(1)|Y(0)=i]$ or $M_{2i}=\mode[Y(1)|Y(0)=i]$. 
These are more
detailed versions of the two dimensional estimands described above. 
It is important to note
that any function of the conditional distribution of the 
potential outcomes can be used as a multidimensional estimand.
Conditional estimands provide a way of measuring the 
magnitude of the effect relative to treatment when no numeric scale is available. 
Since these estimands depend on the joint distribution of the 
potential outcomes, estimation requires modeling assumptions
about the potential outcomes. 
The example in Section~\ref{dataexamp} illustrates how conditional
estimands provide a meaningful summary for the effect of
parental education on a child's education achieves, which cannot
be quantified by considering unconditional estimands.  

Several other multidimensional estimands have been
proposed in the econometrics literature. \cite{boes2013nonparametric} proposes 
multidimensional estimands that describe differences in the distributions under
the control and treatment (rather than differences in the potential outcomes
themselves). For example, Boes defines 
$\Delta^{\rm ate}_i=P_1(i)-P_0(i)$ and 
$\Delta^{\rm SO}_i=\sum_{j\leq i}(P_1(j)-P_0(j))$. 
Similar effects are discussed in \cite{li2008bayesian} in the
context of a specific Bayesian model for ordered data.
If one observes $\Delta^{\rm ate}_i\geq 0$ for all $i$ or
$\Delta^{\rm SO}_i\geq 0$ for all $i$, it suggests that the potential
outcome under treatment is stochastically greater than the potential
outcome under control. However, if the $\Delta^{\rm ate}_i$ are sometimes
greater than $0$ and sometimes less than $0$, this estimand appears to carry little 
information unless particular meaning is available for
a single level of $\Delta^{\rm ate}_i$ or $\Delta^{\rm SO}_i$. 
If the goal of the estimand is to capture the difference
between the distributions under control and treatment, the one-dimensional distance
estimand proposed above can do that with a scalar summary. 
These estimands  
differ significantly from the estimands we proposed
because we are interested
in conditional statements that are meaningful for each level $i$ individually,
whereas the unconditional statements of \cite{boes2013nonparametric}
most often must be presented for all levels $i$ simultaneously. 
Also note that interpretability of the latter requires the signs of all the related effects to match, an empirical result that cannot be guaranteed for any specific data set.

%

\subsection{Estimands on the latent scale\label{sec:laten}}
When a latent variable formulation of ordinal non-numeric 
potential outcomes is used, causal estimands can be defined
on the latent scale. Recall that we refer to the pair $(Z_i(0),Z_i(1))$
as to the latent potential outcomes for unit $i$ whenever there 
is a function $g(\cdot)$
that maps 
$Y_i(1)=g(Z_i(1))$ and
$Y_i(0)=g(Z_i(0))$. 
If the function $g$ is fully
identified then the continuous latent potential outcomes can be used
as de-facto outcomes, and causal analyses can leverage  classical results from  the literature about continuous potential outcomes \citep{rubin2005causal}. In particular,
in this case, the average treatment effect on latent potential outcomes, $\E[Z(1)-Z(0)]$, becomes a meaningful causal effect. However, the 
identifiability of $g$ is key to  estimands defined on the latent scale being meaningful. 
Most often we attempt to infer $g$ from the data, which requires
defining an explicit dependence of the map on the two treatments
yielding functions $g_1$ and $g_0$ for the treated and control
potential outcome maps. These functions likely
have both a location and a scale lack of identifiability, in the 
sense that $\exists h_t$ such that $h_t(cZ_i(t)+b)=g_t(Z_i(t))$. This
non identifiability is critical for the interpretability of causal effects
on the latent scale. In particular, non identifiability of the scale
leads to the following two latent ATEs, under two different
scale assumptions, having the same 
interpretation on the observed scale:
$$\E[Z(1)-Z(0)] \text{ and } \E[5Z(1)-5Z(0)].$$
The latter is five times the size of the former on the latent 
scale, which is undesirable.

The good news is that when this lack of identifiability persists (and
it is unlikely that there is a situation where non identifiability is eliminated in a 
non-artificial way) we can still use the latent scale in order to define
causal effects on the observed scale. For instance, we can consider the estimand
 $\median[g_1(Z(1))|g_0(Z(0))=j]$
as estimand for the conditional median of the
potential outcomes on the observed scale, under treatment, given a particular level
under control. 
While this might appear tautological, the explicit dependence 
on the latent scale and the map $g_t$ is important as it is a statement
about the {science}. 

%% file: tests.tex
\section{Hypothesis testing for ordinal causal effects}\label{sec:test}

In causal analyses a common goal is to
conduct a Fisher exact test for a sharp null hypothesis.
In the classical setting, the sharp null hypothesis of
constant treatment effect can be studied in the same
way as a null of no effect at all. When dealing with
ordinal non-numeric data, however, these two cases need to be analyzed separately. 
For the sharp null of no effect, we construct a Fisher
exact test as in the classical literature, while 
the testing for constant non zero effects requires a more 
precise definition of the hypothesis and leads to 
a permutation test that is not exact as the null is
composite. 

\subsection{Testing for no effect\label{sec:noeff}}

The null hypothesis of no individual level effect 
is the same for ordinal non-numeric data as it is for
numeric data. We define the sharp null
of no effect as $H_0^{\rm no\ effect}:Y_i(0)=Y_i(1)$.
That is,
under this null, an individual's potential outcomes
are equal and so by observing one of the potential
outcomes, we observe the complete
{science} table. This leads to the following construction
of a randomization distribution for the observed data $\{(Y_i^{\obs},W_i):i\leq m\}$
where $W_i$ is the observed treatment status and $Y_i^{\obs}$ is the observed
outcome:
\begin{enumerate}
\item Let $\Pi$ be the collection of all permutations of 
the integers $1$ to $m$.
\item For $\pi\in\Pi$ compute a test statistic $T_\pi$ based on
$\{(Y_i^{\obs},W_{\pi i}):i\leq m\}$.
\item Calculate $\Pr(T>T_{\rm ident})$ where $T$ is distributed according
to the empirical distributions of $T_{\pi}$.
\end{enumerate}

In classical settings, the test statistic $T$ is frequently chosen to be
the average treatment effect. Since this is not meaningful
for ordinal non-numeric outcomes, a different one dimensional 
statistic must be chosen. In particular, it is reasonable
to consider
a test statistic such as $T=d(P_1,P_0)$ for
$d(\cdot,\cdot)$ some measure of distance such as total variation.

\paragraph{Example:} In Section~\ref{dataexamp}
we analyze in detail data from the 1994 General
Social Survey on educational outcomes of children
whose parental education was either below college or
above college \citep{davis1991general}. The distributions of the potential
outcomes under control and under treatment are
presented in Figure \ref{fig:margins}. A visual
inspection of the two distribution suggests that 
they are different. We make this concrete by 
performing the test for the sharp null of no 
effect.
The test statistic we use is the total variation
distance between the 
marginal distributions of the potential outcomes.
The range of the null distribution based on 10000
permutations of the treatment assignment 
is $(0.174,0.413)$ while the
observed value of the test statistic is $0.804$,
giving a $p$-value of $0$.

\begin{figure}[ht]
	\centering
	\includegraphics[width=0.90\textwidth]{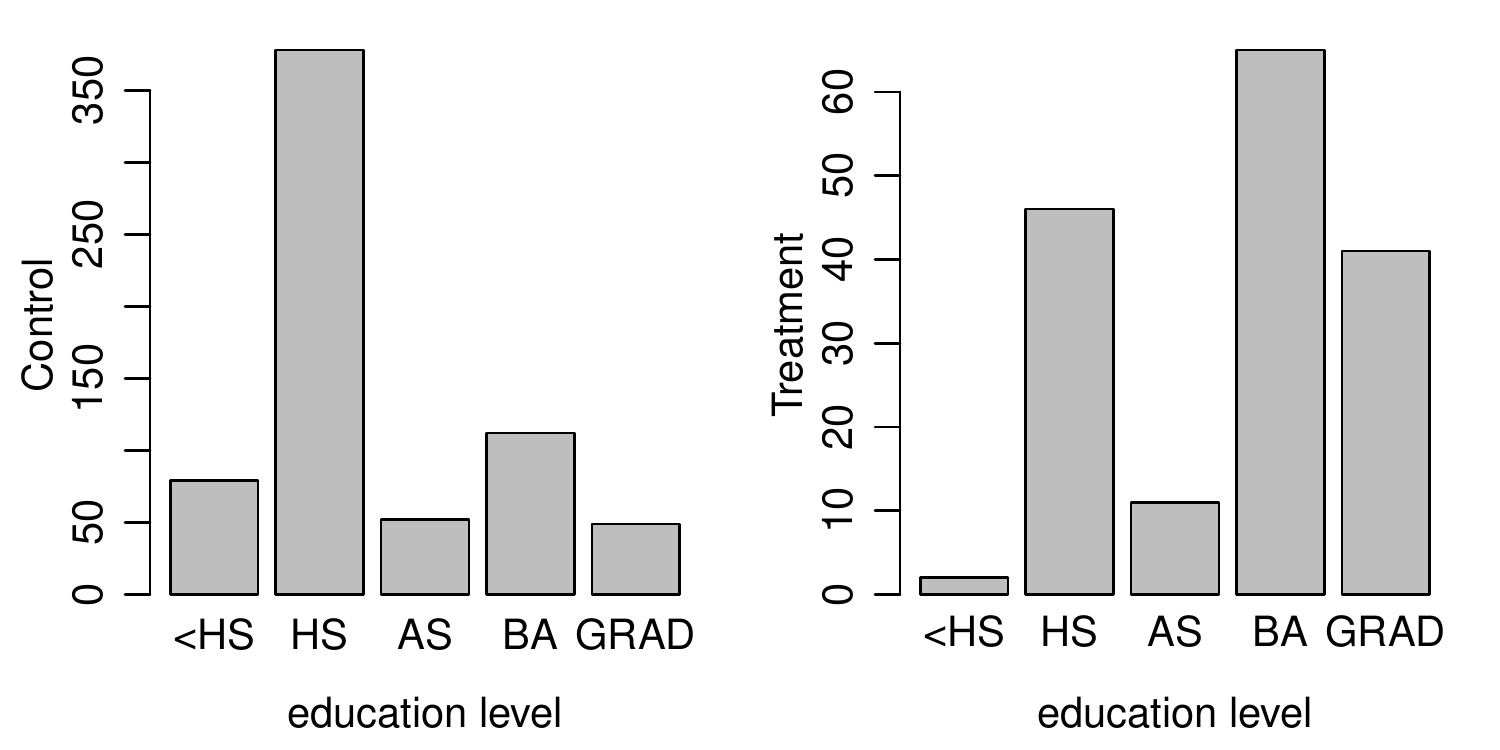}
	\caption{Marginal distributions of outcomes
	under control and under treatment for educational
	outcome data from the 1994 General Social Survey. 
	``$<$HS''=``less than high school'', ``HS''=``high school diploma'',
	``AS''=``associates'', ``BA''=``bachelor degree'', ``GRAD''=``graduate degree''.}
	\label{fig:margins}
\end{figure}
\subsection{Testing for non zero effects\label{sec:const}}

As previously discussed, since the responses
we consider are ordinal non-numeric, the distance
between consecutive categories cannot be assumed equal. As such,
the sharp null of ``constant effect'' that's normally denoted by
$H_0^{\rm const}:Y_i(1)-Y_i(0)=c >0$ is no longer meaningful 
as it would immediately assign a numeric scale to the data.
Because of the lack of scale, it is tempting to 
phrase a constant effect null hypothesis conditionally as
$H_0^{\rm const}:Y_i(1)-Y_i(0)=c$ such that $c>0$ only if $Y_i(0)=j$ meaning that
the effect of treatment is a constant $c$ categories for individuals
with control potential outcome equal to $j$ and zero otherwise. But such a conditional statement
does not fully specify the null space. Without loss of generality, let
$c=1$ and $j=1$ (the lowest category); that is, under the null, the effect of treatment on individuals
in the lowest category, under control, is one category exactly. Unfortunately,
if we observe $Y_i^{\obs}=1$ for $W_i=1$, it is not possible to construct the full
{science} table since it is unclear what to set $Y_i^{\mis}$.
This leads to the following definition of what is arguably the simplest null hypothesis of  non zero effect.
\begin{definition}\label{prop:basic}
For a $k$ level ordinal outcome, fixing $j<k$, the simplest non zero effect null hypothesis is given by
\begin{itemize}
\item For $Y_i(0)=j$, $Y_i(1)-Y_i(0)=0$ with probability $1-p$.
\item For $Y_i(0)=j$, $Y_i(1)-Y_i(0)=1$ with probability $p$.
\item Marginal probabilities: $\Pr(Y_i(0)=j)=l_{j}$ 
and $\Pr(Y_i(0)=j+1)=l_{j+1}$.
\item For $Y_i(0)\neq j$, $Y_i(1)-Y_i(0)=0$.
\end{itemize}
\end{definition}
Note that the statement ``with probability $p$'' can refer to a superpopulation
quantity, or simply to the finite population proportion of individuals
whose science table contains the requisite potential outcomes.

A scenario when such a null hypothesis is meaningful comes up
in follow-up studies, where preliminary results suggest a nonzero
effect for certain groups, but do not have any information about other
groups. This is the simplest null in the sense that one cannot 
formulate a 
nonzero null for ordinal non-numeric data without specifying 
at least this many conditionals and marginals 
or imposing a numeric scale. 
The above null can also
be stated as a condition on the potential outcome $Y_i(1)$ under
treatment by replacing ``For $Y_i(0)=j\dots$'' with ``For $Y_i(1)=j+1\dots$''
and replacing $p$ with $1-p$.
The statement of the null in Definition~\ref{prop:basic} also dictates 
the procedure by which one can test the null. Fixing $j,p$ 
and $l_j,l_{j+1}$ we can
complete the {science} table in the following way:
	\begin{itemize}
	\item If $W_i=0$ and $Y_i^{\obs}=j$ then $Y_i^\mis=j+1$ with probability $p$
	or $Y_i^\mis=j$ with probability $1-p$. 
	\item If $W_i=1$ and $Y_i^{\obs}=j+1$ then set $Y_i^{\mis}=j$ with probability $q=pl_{j}/(pl_{j}+l_{j+1})$
	or $Y_i^{\mis}=j+1$ with probability $1-q$.
	\item Else, let $Y_i^\mis=Y_i^\obs$.
	\end{itemize}
A permutation test can then be performed using the
complete {science} table. Because the null is composite,
multiple science tables must be computed and the distribution
of the test statistic constructed over all of them.
It is important to note that the quantity of interest in this
setting is the conditional probability $p$, suggesting
that during the testing procedure the marginals are nuisance
parameters that can either be integrated over,
or potentially set to the marginals of the observed
data.

A more general null hypothesis
that has an interpretation as a constant non zero effect can be motivated by a
latent variable formulation of ordinal data. Recall that,  if the 
potential outcomes have a latent representation, then for the continuous latent
potential outcomes $(Z_i(0),Z_i(1))$ one could test for a 
constant treatment effect of 
$H_{0,\ {\rm latent}}^{\rm constant}:Z_i(1)-Z_i(0)=c$.
Under certain conditions, an interpretation of this null is available in terms of the null
specified by Definition~\ref{prop:basic}. The latent potential outcomes are mapped
back to the observed scale by the function $g$. Here the explicit dependence
of $g$ on the treatment is suppressed since we assume an
additive treatment effect, that is $g_1(\cdot)=g_0(\cdot-c)$. 
Let $c>0$ be such that for all 
levels $j$ and for all latent values $z$ for which $g(z)=j$ we have
$g(z+c)\leq j+1$, with equality for some $z$. That is, for $c$ small enough, a constant latent effect
is interpretable as at most improvement by one category on the observed scale
for all individuals but with different probabilities.
That is, we would have a combination of the simple nulls of
Definition~\ref{prop:basic} where for each level $j$, the conditional
probability that $Y_i(1)-Y_i(0)=0$ for $Y_i(0)=j$ is given by $p^{(j)}=\frac{|z_j|-c}{|z_j|}$ where
$z_j=\{z:g(z)=j\}$.

For a general $c$, the interpretation of the test on the observed scale
becomes more complicated but involves a similar structure to the 
one in Definition~\ref{prop:basic}: WLOG let $c>0$ be as before, but for a single  
category $j$, for all $z$ with $g(z)=j$ we have $g(z+c)\leq j+2$ with equality for some $z$. 
As such, the improvement is at most two levels for individuals for whom the
potential outcome under control is $j$. As such, a 
null hypothesis that corresponds to this on the observed scale
requires defining probabilities $p_0,p_1,p_2$ for $Y_i(1)-Y_i(0)\in\{0,1,2\}$ when $Y_i(0)=j$.
In particular, $p^{(j)}_0=\frac{|z_{j}|-c}{|z_{j}|}$, $p^{(j)}_1=\frac{|z_{j+1}|}{|z_{j}|}$
and $p_2^{(j)}=1-p^{(j)}_0-p_1^{(j)}$.

\subsection{Example and fiducial type intervals} 
\label{sub:example}
In this Section, we consider a (simulated) randomized experiment
with 500 units and an ordinal non-numeric outcome
labeled $\{1,2,3\}$, where label 1 indicated control. The joint distribution of 
the potential outcomes for the simulated data is given by 
Table \ref{tab:joint}. This type of joint represents
information that the treatment leads to at most
a change of plus one category with possibly different
probabilities of change conditional on the potential
outcome under control. The estimands of interest are
the conditional probabilities $q_1,q_2$, and it is easy to see
that they both fit into the paradigm of 
Proposition~\ref{prop:sep} of the next section and so they 
can be estimated without specifying an additional model for the data. For example
$1-\hat{q}_1={\rm mean}(Y(1)=1)/{\rm mean}(Y(0)=1)$ where
each of the means is an entry in the empirical
distribution of the marginals. The estimates are 
truncated at 0 and 1 to get valid probabilities, as it is commonly done for correcting method of moments estimators,
though this correction is rarely necessary in the large sample setting we consider. 

\begin{table}
	\caption{Joint distribution of the potential outcomes.
	\label{tab:joint}}
	\centering
	\input{joint_test.tex}
\end{table} 

The null hypothesis described in the previous
section requires specifying values $(\eta_1,\eta_2)\in(0,1)^2$
that correspond to the conditional probability of a positive
one category change when the potential outcome under control is
one and two, respectively, and parameters $(\nu_1,\nu_2,\nu_3)$
in the simplex that correspond to the marginals. 
Let $p(\eta_i)$ be the
$p$-value for the conditional probability $q_i$ in 
Table~\ref{tab:joint}.
Performing the test for
multiple values of $\eta_i$ and $\nu_i$
recovers different $p$-values. This suggests that we can construct
$100(1-\alpha)\%$
fiducial type intervals for the conditional effects $q_i$ with
the following bounds:
 $\xi_i^L=\sup\{\eta_i:p(\eta_i)\leq\alpha/2\}$ and
 $\xi_i^U=\inf\{\eta_i:p(\eta_i)\geq 1-\alpha/2\}$
 \citep[e.g., see][]{wang2000fiducial,dasgupta2014causal}. 
Since the nuisance parameters $(\nu_1,\nu_2,\nu_3)$ are not
known, we must also consider them in the sequence of nulls.
As such, by projecting the $p$-values down to the space of $(\eta_1,\eta_2)$
the intervals we recover are conservative. This procedure is
superior to using a plug-in estimate for the nuisance
parameters as the intervals using that procedure can be either 
conservative or anti-conservative \citep{imbens2014causal}.

Here we simulate $N=500$ experimental units
from the joint distribution
in Table~\ref{tab:joint} with parameters
$(c_1,c_2,c_3)=(1,1,1)/3$
and $(q_1,q_2)=(7/10,2/3)$. 
From the observed outcomes, the estimates are
$\hat q_1=0.73$ and $\hat q_2=0.61$.
We perform the test for
null hypothesis where $(\eta_1,\eta_2)$ are 
on a uniform $30\times 30$ grid in $[0.1,0.999]^2$
and $(\nu_1,\nu_2,\nu_3)$ are sampled uniformly
on the simplex, but restricted to be no bigger than 
$0.6$ and no smaller than $0.15$ each. 
For each 
null, we perform the test as described in 
Definition~\ref{prop:basic}, constructing 1000 null
science tables and getting 100 randomizations from each one.

Figure~\ref{fig:dens} provides the null densities for a 
particular null
for $q_1$ and $q_2$. The dashed vertical lines are the
observed values having $p$-values $0$ and $0.62$ respectively.
Figure~\ref{fig:pvals} provides 
the $p$-values for the $q_i$ as a function of
the $\eta_i$. The vertical lines provide the two
$95\%$ fiducial type intervals: $[0.63,83]$ and $[0.56,0.83]$
respectively.

\begin{figure}[!ht]
	\centering
	\includegraphics[width=0.9\textwidth]{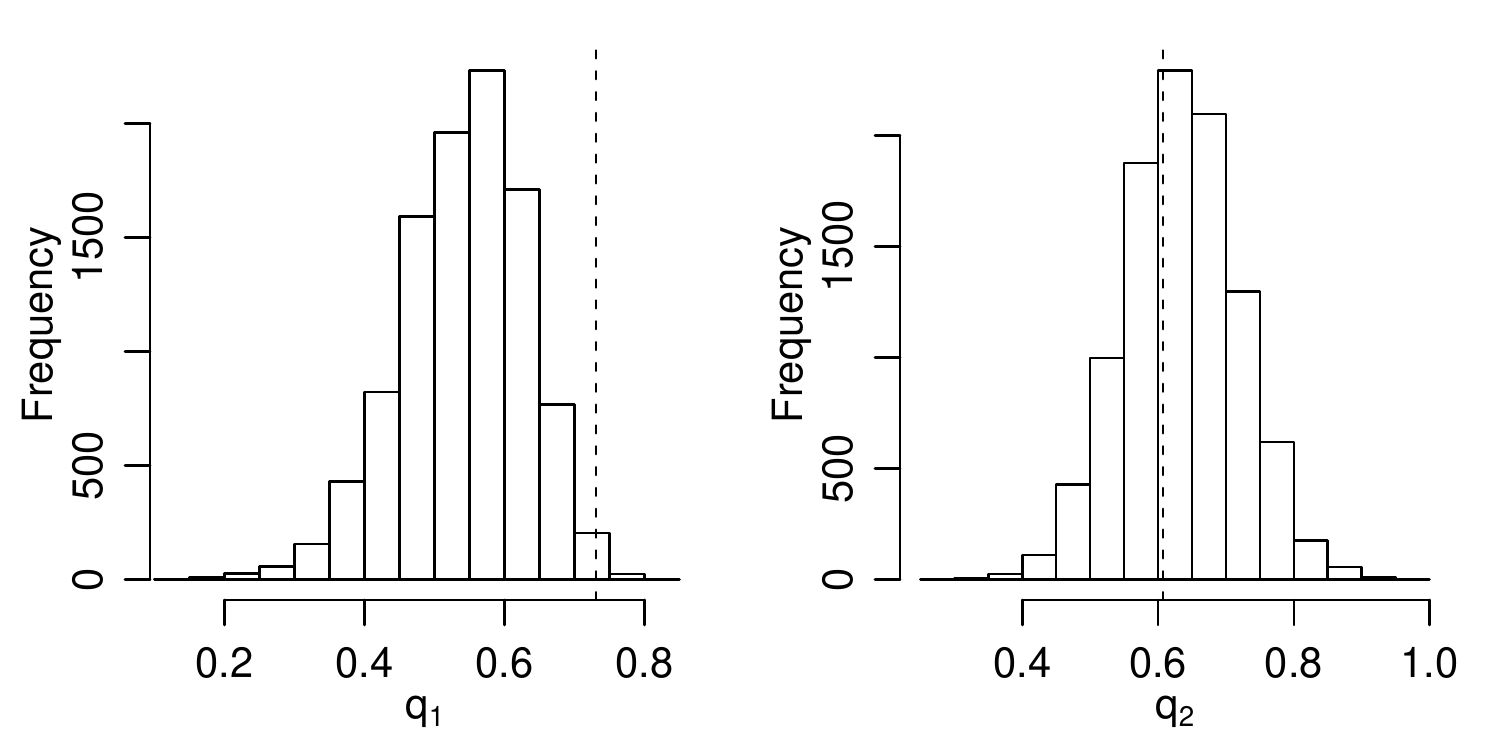}
	\caption{Randomization distributions for $q_1$ (left panel) and $q_2$ (right panel)
	for parameter values $(\eta_1,\eta_2)=(0.487, 0.624)$ 
	and $(\nu_1,\nu_2,\nu_3)=(0.280, 0.549, 0.171)$.}
	\label{fig:dens}
\end{figure}
\begin{figure}[!t]
	\centering
	\includegraphics[width=0.9\textwidth]{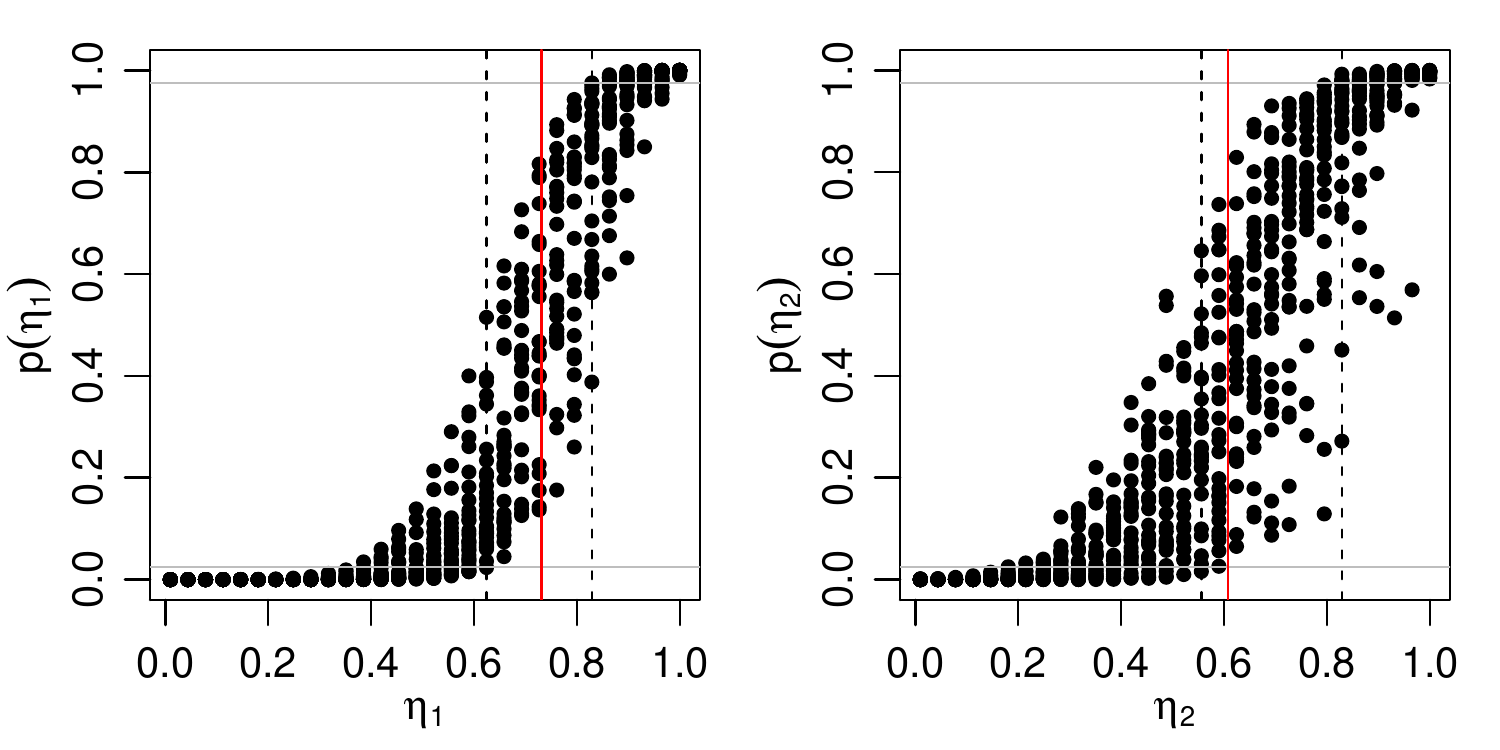}
	\caption{$p$-values as functions of parameters $\eta_1$ and $\eta_2$
	with the dashed lines representing $95\%$ conservative fiducial intervals
	and the red lines representing the observed data estimates.}
	\label{fig:pvals}
\end{figure}

%% file: joint_test.tex
\centering
\begin{tabular}{r|rrr}
  &\multicolumn{3}{c}{$Y(1)$}\\
 & 1 & 2 & 3 \\ 
  \hline
1 & $(1-q_1)c_1$ & $q_1c_1$ & 0 \\ 
 $Y(0)$ 2 & 0 & $(1-q_2)c_2$ & $q_2c_2$ \\ 
  3 & 0 & 0 & $c_3$ \\ 
   \hline
\end{tabular}

%% file: estimation.tex
\section{Inference and approaches to identification of causal effects}\label{sec:estimation}
Estimation of causal effects has been discussed in detail for
many types of data \citep{rubin1978bayesian,rubin1974estimating,rubin2005causal}, 
but no explicit discussion is available in the statistics 
literature for ordinal non-numeric data. Throughout we will assume the standard assumptions of the Rubin Causal Model, such as the stable unit treatment value assumption (SUTVA), ignorability of the assignment mechanism that leads to the realized outcomes on the observed scale, and (whenever relevant) of the assignment mechanism that leads to the realized outcomes on the latent scale.
In such situations, estimands that are meaningful 
in both ordinal non-numeric and continuous outcome settings can be estimated
using existing machinery \citep{imbens2014causal}. In particular, when operating on the 
observed scale
without assuming a model for the potential outcomes,
we can estimate any estimand that can be written explicitly as a 
difference between a function of the distribution of the 
responses under treatment and the distribution of the responses
under control. In the classical setting of continuous outcomes
this translates to the average treatment effect being estimable 
since expectations are linear (and so $\E[Y(1)-Y(0)]=\E[Y(1)]-\E[Y(0)]$). 

For ordinal non-numeric data, however, an estimand that directly compares outcomes under control and under treatment without including a model for the science does not exist.

\begin{proposition}\label{prop:sep}
For ordinal non-numeric data, an estimand that 
cannot be written as a combination of
separated functions such as $$f(Y(0),Y(1))=f_1(Y(1))-f_0(Y(0))$$
cannot be estimated without a model for the science.
For any estimand that can be separated as above, $f_1$
is a function of the distribution of $Y(1)$ alone and $f_0$
is a function of the distribution of $Y(0)$ alone. If there exist
unbiased estimates
of $f_0$ and $f_1$, then there is an unbiased estimate of
the estimand of interest.
\end{proposition}
The proof of Proposition \ref{prop:sep} is a direct
consequence of classical results about estimation
of causal effects \citep{imbens2014causal}. Intuitively, if an estimand cannot be written
as a combination of functions that only depend
on the marginal distributions of the potential outcomes
then it necessarily depends on their joint distribution. 
Estimation in this case requires a model for the science.
If an estimand can be separated as such and if unbiased
estimates exist of each part, then any combination of unbiased
estimates remains unbiased.

Proposition \ref{prop:sep} unsurprisingly states that 
``simple'' estimands can be estimated as they previously have. 
However, it quickly emerges that for ordinal non-numeric data
the estimands of greatest interest do not satisfy the separation of 
Proposition \ref{prop:sep}. For example, the conditional medians
$\median[Y(1)|Y(0)=i]$ described in Section \ref{sec:estimand}
do not satisfy the condition. To estimate these effects we require additional 
assumptions about the science. The first assumption we consider
is that of monotone treatment
effects which assists in 
the identifying the joint distribution of the potential outcomes. 
While it is appealing, this assumption does not fully resolve the 
issues with estimating the estimands of interest. As such, we 
develop estimation under the latent variable assumption---explicitly describing the functional
relationship between the potential outcomes on the 
observed and latent scales.

%
%

\subsection{Monotone treatment effect}
Formally, the assumption of a monotone treatment effect 
states that the potential outcome under treatment is at least
as large as that under control. That is, $Y_i(1)\geq Y_i(0)~\forall i$.

We first consider the case of a binary treatment
and ordered binary outcome. If we can assume monotonicity
of treatment then we are able to recover the full joint
distribution from the two marginals and hence any of the estimands previously
listed.
Monotonicity of treatment gives us that $\Pr(Y(0)=1,Y(1)=0)=0$.
This immediately implies that $\Pr(Y(0)=0,Y(1)=0)=\Pr(Y(1)=0)$.
We can just as easily recover the other elements 
of the joint distribution:
\begin{eqnarray*}
 \Pr(Y(0)=1,Y(1)=1) & = & \Pr(Y(1)=1|Y(0)=1) \cdot \Pr(Y(0)=1) \\
                    & = & \Pr(Y(0)=1)
\end{eqnarray*}
where the second equality is due to the conditional
probability equaling 1 (by monotonicity).
The final probability can be recovered via
the additivity of probabilities to get
$\Pr(Y(0)=0,Y(1)=1)=1-\Pr(Y(1)=0)-\Pr(Y(0)=1)$.

In cases where the outcome has more than two levels, 
monotonicity is not sufficient for fully identifying
the joint probability distribution of the potential outcomes.
However, certain estimands of interest can be bounded
by estimands that do not require knowledge of the full joint. 
\begin{proposition}
For $k$ level ordinal non-numeric outcomes under the assumption 
of a monotone treatment effect, we have
\begin{enumerate}
	\item\label{bnd1} $j\leq\median[Y(1)|Y(0)=j]$
	\item\label{bnd2} $j\leq\median[Y(1)|Y(0)\geq j]$
	\item\label{bnd3} $\median[Y(0)|Y(0)\leq j]\leq\median[Y(1)|Y(0)\leq j]$
\end{enumerate}
\end{proposition}
\begin{proof}
Monotonicity states that $\Pr(Y(1)\geq Y(0))=1$,
which implies $\Pr(Y(1)\geq i|Y(0)=i)=1\forall i$. 
This proves \ref{bnd1} and \ref{bnd2}. Inequality 
\ref{bnd3} follows
from the fact that the truncated distribution $Y(1)|Y(0)\leq j$
is necessarily shifted to the right of the
truncated distribution $Y(0)|Y(0)\leq j$.
\end{proof}
Similar bounds can be derived for statements about other
conditional functions such as the mode or other quantiles.
This assumption is only appropriate when the non-negativity
or non-positivity of the treatment effect is known \emph{a priori}. 
For example, it is reasonable to assume that a reading comprehension
program can only improve reading skills, but it might be inappropriate 
to assume that attending a classical music concert cannot have both
positive and negative effects on a person's state of mind.

%
%

\subsection{Latent variable formulation}
We have previously introduced the latent variable
formulation for ordinal non-numeric potential outcomes. 
Intuitively, this approach makes an assumption about the
underlying {science}. This is because we require an
explicit functional
relationship between the potential outcomes
proper, $(Y_i(0),Y_i(1))$,
and the latent potential outcomes, $(Z_i(0),Z_i(1))$. We 
write this map as a pair of functions $g_t$ for $t\in\{0,1\}$ 
that map $Z_i(t)\rightarrow Y_i(t)$.
We make the dependence on the treatment explicit since estimation
cannot proceed without identifying what the form of treatment
is on the latent scale. 
The only functional restriction on the two $g_t$ functions is that
$y=g_t(z)\geq g_t(z^\prime)=y^\prime$ implies that $z\geq z^\prime$. 
This general formulation is not a model for estimation but rather
a fundamental assumption about the science. 

For estimation to work, we must 
choose an explicit functional form for the treatment effect. 
A natural choice
here is a linear treatment effect, thus simplifying the notation
as we have $g(Z(1))=g(Z(0)+c)$.
A possible nonlinear functional form for the treatment effect
makes the assumption that $g_0$ and $g_1$ have different 
cutoff values for mapping between the latent and observed scale. 
Whatever the assumption made about $g$, the procedure for 
inference is as follows.
\begin{enumerate}
\item\label{les1} Choose the functional form for $g_0$ and $g_1$.
\item\label{les2} Write explicitly $Y_i^{\obs}=g_{W_i}(Z_i^{\obs})$ and 
let $Z_i^{\obs}=f(X_i)$ 
where $W_i$ is the assigned treatment to unit $i$. In the
language of generalized linear models
 $g$ is our link function and $f$ the mean function that
is there for notational purposes to explicitly state the 
functional dependence of the latent variable on additional covariates. 
\item\label{les3} Estimate $\hat g_0$, $\hat g_1$ and $\hat f$ using the observed data: 
$(Y_i^{\obs},W_i,X_i)$.
\item\label{les4} Impute the missing potential outcomes $Y_i^{\mis}$
by writing 
$\hat Z_i^{\mis}=\hat f(X_i)$
and let $Y_i^{\mis}=\hat g_{1-W_i}(\hat Z_i^{\mis})$.
\item\label{les5} Estimate the estimand on the observed scale using the 
imputed data $(Y^{\obs},Y^{\mis})$.
\end{enumerate}
The procedure described above takes advantage of the
continuity of the latent scale to make inference about
the observed scale tractable. When we consider a 
linear effect of treatment on the latent scale, we can
write in steps \ref{les2} and \ref{les4} above: $Z_i^{\obs}=f(X_i,W_i)$ and drop the 
dependence of $g$ on the treatment assignment. 
One of the biggest advantages for employing the latent 
scale formulation is the ability to estimate the uncertainty
about the observed scale estimands. To do so we must
compute the variance of the predictive sampling distribution of
the observed scale potential outcomes. While this can be done
explicitly as outlined above, the desired quantities can be conceptualized as a consequence of a Bayesian estimation approach in which the estimands on the 
observed scale are functions of the posterior predictive
distribution. We outline such a
Bayesian estimation procedure next, and we discuss the selection of
priors.

%% file: model.tex
\section{Bayesian formulation}\label{sec:bayes}
Bayesian estimation procedures play an important role
in causal analysis \citep{rubin1978bayesian}. We  consider a single treatment and single
control group, with $k$ ordinal non-numeric outcomes. For each unit $i$ we  observe a potential
outcome associated with its treatment assignment as well as 
covariates that provide background information and will be used
to adjust the outcomes. As in Section \ref{sec:estimation}, we 
 write $Z_i(t)$ for the latent scale potential outcome
for treatment $t$, and we need to estimate the functions
$g_t$ and $f$. 
The potential outcomes on the observed scale, $Y_i(t)$, take on values
in $\{1,\dots,k\}$. There are $n$ units.

Formally, we can write the assumed model as follows.
\begin{align*}
Y_i(t)&=g_t(Z_i(t))\\
Z_i(t)&=f(X_i,W_i) + \epsilon_i(t)\\
\epsilon_i(t)&\overset{iid}{\sim}\pi
\end{align*}
where $\pi$ is the distribution of the errors that has
no unknown parameters (this is an identification assumption
and the reason why latent scale estimands are frequently
inappropriate),
and where $Y_i(t)=g_t(Z_i(t))=j$ if $z\in(s_{j-1}^t,s_j^t)$ 
for a monotonic increasing sequence 
$\mathcal{S}^t=\{-\infty=s_0^t,s_1^t,\dots,s_k^t=\infty\}$. 
For example, if $g_t=g$, $f(X_i,W_i)=X_i\beta+W_i\beta_w$,
$\mathcal{S}^1=\mathcal{S}^0$ and $\pi$ is a standard normal distribution, then this is a standard order probit model. 
If $\pi$ is the logistic distribution then this is a
standard ordered logit, and so on. One can also consider 
a scenario where $\beta_w=0$ but $g_1\neq g_0$ and
$\mathcal{S}^1\neq\mathcal{S}^0$; that is,  no linear effect of treatment 
on the latent variables, but possibly different cutoff values
for the treated and control groups. All of these options 
are estimable using a Bayesian approach as long
as the distribution of the missing
potential outcomes conditional on the observed ones is tractable.
Below we outline a standard Gibbs sampler scheme for the
ordered probit model.  

 
\subsection{Prior choice for the ordered probit}
The first step in Bayesian inference is prior choice. 
We require priors for the parameters $(\beta,\beta_w)$, as
well as for the cutoff values $\mathcal{S}$. In a slight
abuse of notation, we will write $X_i$ for the combined vector
of pre-treatment variables with the indicator for treatment and we 
will write $\beta$ for the combined vector of coefficients. 
Since the latent errors are standard Gaussian, a natural conjugate
prior for the parameters $\beta$ is also Gaussian, for example, with
mean 0 and covariance matrix $n(X^tX)^{-1}$. 
The more complicated prior choice is for the cutoff 
values $\mathcal{S}$. Note that, if we knew 
the cutoff values exactly, we would in fact have the proper
scaling for doing inference on the latent scale. 

Even though the cutoff parameters are critical for  identifying the model,
 a default prior for the cutoff 
values $\mathcal{S}$  does not exist. We note that a prior for the cutoff values does not have to 
respect the ordering of the elements of $\mathcal{S}$
since the ordering will be imposed by the likelihood function.
As such, without scientific knowledge,
a reasonable prior that carries very little
information is essentially flat. A potential choice is a product of mean zero
normal variables with a large variance parameter. 

Because of the difficulty of choosing a reasonable prior
for $\mathcal{S}$ and since we are not interested in 
interpreting the latent scale of the potential outcomes,
we argue it is more meaningful to consider the rank likelihood
as the data likelihood as discussed below. 
 
\subsection{Posterior inference for the ordered probit}
For the probit model with independent priors for
thresholds $\mathcal{S}$ and a $N(0,n(X^tX)^{-1})$ prior
for the $\beta$s the full conditional distributions
are easy to derive and so we can formulate a Gibbs sampler
with the following steps.
\begin{enumerate}
	\item\label{posp1} Initialize $\beta_{[0]},Z_{[0]},\mathcal{S}_{[0]}$ 
	where in an abuse of notation, the bracketed subscript
	refers to the iteration of the Markov chain Monte Carlo 
	algorithm.
	\item\label{posp2} Sample $ \beta_{[l]}|Z_{[l-1]},Y,\mathcal{S}_{[l-1]}\sim N(\frac{n}{n+1}(X^tX)^{-1}X^tZ_{[l-1]},\frac{n}{n+1}(X^tX)^{-1})$.
	\item\label{posp3} For each unit $i$ sample $Z_{[l]i}|\beta_{[l]},Y,\mathcal{S}_{[l-1]}\sim N(X_i\beta_{[l]},1)\delta_{s_{[l-1]Y_i-1},s_{[l-1]Y_i}}(Z_{[l]i})$.
	\item\label{posp4} For each cutoff point $j$ sample $s_{[l]j}|Y\sim p(s_{[l]})\delta_{\max\{Z_{[l]i}|Y_i=j\},\min\{Z_{[l]i}|Y_i=j+1\}}(s_{[l]j})$
	\item\label{posp5} Sample $Y^{\mis}|Y,\beta_{[l]},\mathcal{S}_{[l]}$ from
	the posterior predictive distribution 
	and construct the causal estimand of interest $T_{[l]}$. 
\end{enumerate}
Above, the function $\delta_{a,b}(c)$ is equal to one if $a\leq c\leq b$. Steps \ref{posp2} through \ref{posp5} are iterated until the joint distribution
of $(\beta,Z,\mathcal{S})$ reaches stationarity. The  
distribution of the $T_{[l]}$ for $l\in\{1,\cdots,L\}$ provides an
approximation for the posterior predictive distribution of the
causal estimand of interest on the observed scale. A histogram summary
of the posterior predictive distribution provides both a 
measurement of the most likely value of the estimand and a measure
of the certainty about this value.

If we choose to avoid the prior specification for $\mathcal{S}$,
we can employ the rank likelihood, first discussed by \cite{pettitt1982inference} and employed in the ordered
probit setting by \cite{hoff2008rank}.
In the rank likelihood case, we do not need to estimate the cutoff
values in $\mathcal{S}$. Instead we require that
the latent outcomes $Z$s must lie in the set
$R(Y)=\{Z\in\mathbb{R}^n:Z_{i_1}<Z_{i_2}\text{ if }Y_{i_1}<Y_{i_2}\}$.
 Posterior inference with this assumption forgoes
 step \ref{posp4} of the procedure above. The full conditional for the 
 parameters $\beta$ remains the same and so only step \ref{posp3}
 must be changed to reflect the rank likelihood as the 
 full conditional distribution of $Z_{[l]j}$ now depends
 on the remaining $Z_{[l-1]j^\prime}$. That is
 the sampling distribution in step \ref{posp3} becomes
\[p(Z_{[l]j}|\beta_{[l]},Z_{[l]}\in R(Y),Z_{[l-1] -i})\propto N(X_i\beta_{[l]},1)\delta_{\max\{Z_{[l]i}:Y_i<Y_j\},\min\{Z_{[l]i}:Y_j<Y_i\}}(Z_{[l]j}).
\] 

\paragraph{Complications.}
In applications it may be desirable to assume that $g_0\neq g_1$ and that $\mathcal{S}^{0}\neq\mathcal{S}^1$.
Both of the  above approaches above can be employed
 to perform inference under this more complex model. Assuming that the 
 only difference between the two functions $g_t$ is in the 
 cutoff values (that is, the mean function $f$ remains the same, with or without the additive treatment effect) the MCMC
 procedure described above does changed substantially. Specifically, step
 \ref{posp2} remains the same where $Z_{[l-1]}$ includes all of the units. 
 The main changes appear in steps \ref{posp3} and \ref{posp5} 
 (and step \ref{posp4} if it is needed): each of these steps is split into
 an update for the control and an update for the treated groups
 since these groups now have their own parameters. 
 Further variations on the model can also be introduced as long as
 all of the requisite probability distributions can either be
 computed or sampled from.

\section{Analysis of educational outcomes in the General Social Survey}\label{dataexamp}
In this section we consider data from the 1994 General Social Survey
on the educational outcomes of a sample individuals living in the 
United States \citep{davis1991general}. Each of the 835 male 
respondents who were between the ages
of 25 and 60 and in the workforce during 
the survey provided information on their educational outcomes as 
well as information about whether at least one of their
parents had attained a college degree or higher. 
The possible levels of education recorded 
for an individual were ``less than high school'', ``high school'', ``associates'', ``bachelor'' and ``graduate''. 
These data are presented
in Table \ref{tab:dat} and Figure \ref{fig:margins} and were previously studied in 
\cite{hoff2007extending,hoff2009first}. 

\begin{table}
\caption{The distributions of individual educational outcomes
according to treatment status.\label{tab:dat}}
\centering
\begin{tabular}{rrrrrr}
  \hline
 &  ``less than high school''& ``high school''&``associates''& ``bachelor'' & ``graduate'' \\
  \hline
Control &  79 & 378 &  52 & 112 &  49 \\
  Treatment &   2 &  46 &  11 &  65 &  41 \\
   \hline
\end{tabular}
\end{table}
Based on the marginal distributions alone one might suspect
that a positive treatment effect exists as over
$50\%$ of individuals in the treated group achieved an
educational level of college or above, while only
$24\%$ of the control group have this educational level.
More generally, the marginal distribution of the
treatment group stochastically dominates the marginal distribution
of the control group. However, this information does not
demonstrate the magnitude of the effect. As such
we consider the conditional estimands that we previously described.
Here we are interested in the conditional 
median potential outcome
under treatment given a particular level of the potential
outcome under control. This estimand can capture
different magnitudes of the effect for individuals whose 
potential outcomes differ under control. For example, 
we might expect that individuals whose potential outcome
under control is college or better to have a lower effect of
treatment than for individuals whose potential outcome under
control is less than college.

The following specifies a model for the potential
outcomes using a latent variable representation with
correlated potential outcomes:
\begin{align*}
Y_i(T)&=g(Z_i(T))\qquad &\epsilon_i(T)\sim {\rm normal}(0,1)\\
Z_i(T)&=\beta T+\epsilon_i(T)\qquad &{\rm cor}(\epsilon_i(0),\epsilon_i(1))=\rho
\end{align*}
The function $g$ maps the latent variable representation to
the ordered non-numeric space and the parameter $\rho$
captures the dependence among the potential outcomes for
an individual. It is clear that the data contains no 
information about the parameter $\rho$ since we only observe
one of the potential outcomes for each individual \citep{imbens2014causal}. 
When $\rho=1$ then the linear
relationship between the potential outcomes on the latent
scale is exact. On the other hand $\rho=0$ suggests that all
of the effect is captured by the coefficient $\beta$ on the 
latent scale. Other choices of $\rho\in(0,1)$ describe
different measures of positive dependence. As in \cite{dasgupta2014causal},
we treat the correlation parameter $\rho$ as known. We explore
four values of $\rho$: $0.25$, $0.50$, $0.783$ and $1$.
The third value of $\rho$ is chosen via the following heuristic
argument: It is the Frechet-Hoeffding upper bound for the correlation
of two random variables with marginals given by the observed
potential outcomes. Since the correlation of two coarsened
random variables is necessarily not bigger than the correlation
between the two uncoarsened versions, the choice of the 
upper bound in our heuristic reflects that the correlation
between the latent potential outcomes is greater than the correlation 
between the observed scale potential outcomes.

We use the Bayesian approach that employs the rank likelihood
described in the previous section to obtain the 
posterior predictive estimates of 
the estimands of interest. We draw 50,000 posterior
predictive samples for each estimand of interest and report
the posterior median in Table \ref{tab:res}. Intervals are provided where
the posterior is not a point. 

\begin{table}
\caption{Posterior median estimates of the estimands of interest for the
General Social Survey where
$1=$``<high school'', $2=$``high school'', $3=$``associates'', $4=$``bachelor'' and $5=$``graduate''.
$95\%$ confidence intervals are reported as follows: $^\star=(4,5)$, $^\triangle=(2,4)$, $^\square=(2,3)$, $^\circ=(3,4)$, otherwise
the interval is a point.\label{tab:res}}
\centering
	{\begin{tabular}{c|p{2cm}p{2cm}p{2cm}p{2cm}p{2cm}}
		&\multicolumn{5}{c}{$\median~[~Y(1) \mid Y(0)=j~]$}\\
		$\rho$ & $j=1$ & $j=2$ & $j=3$ & $j=4$ & $j=5$ \\ 
		\hline
		$0.25$    & $4^\triangle$ & 4 & $4^\star$ & $4^\star$ & $5^\star$ \\ 
		$0.50$    & $2^\square$ & 4 & $4^\star$ & $5^\star$ & 5\\
		$0.783$ & 2 & 4 & $5^\star$& 5 & 5 \\ 
		$1$     & 2 & $4^\circ$& 5 & 5 & 5 \\ 
		\hline
		\end{tabular}}
\end{table}
The results presented in Table \ref{tab:res} 
reveal that the analysis is only mildly sensitive to the 
choice of $\rho$. The conditional estimates that
change the most with $\rho$ correspond to 
categories for which there is not much data information
(``$<$high school'' and ``associates''). For example,
when $\rho=0.25$ we estimate $\median[Y(1)|Y(0)=\text{``<high school''}]$
to be ``bachelor'', but the $95\%$ posterior interval
includes the estimate based on all other $\rho$ values, ``high school''.
The first order conclusion is that in fact there is a treatment
effect, which agrees with previous insights into the subject
of the effect of parental education on child educational
outcomes \citep{burnhill1990parental}. The contribution of our
method provides a breakdown of this effect conditional on
the potential outcome under control. In particular, we see
that for all individuals who would have at least attained
a high school diploma under control, 
the effect of a parental college degree is that they themselves
attain at least a college degree. 

%% file: disc.tex
\section{Concluding remarks}
In this article we described technical difficulties that arise in  causal analyses when the potential outcomes
take ordinal non-numeric values, and proposed solutions. 

In Section \ref{sec:estimand}, 
we proposed a class of 
multidimensional estimands that depend on the
distribution of the potential outcomes under treatment
conditional on those under control. These estimands,
of the form $\median[Y(1)|Y(0)=j]$ are especially useful
 when experiments are conducted with the goal of planning future interventions. For example, in a health experiment where outcomes attempt to measure happiness and can take levels of ``sad'' and ``happy'',
a practitioner can decide to assign an anti-depessant to individuals 
only
if $\median[Y(1)|Y(0)=\text{``sad''}]=\text{``happy''}$. One can imaging such a prescription would be made conditional on covariates.
Additionally,
we described why classical one dimensional estimands,
such as the average treatment effect and the difference
in medians between treatment and control,
are inappropriate for ordinal non-numeric data. One
dimensional estimands that are appropriate must be 
scale-free. When one dimensional estimands are of interest, we recommend using a measure of distance between
the two marginal distributions of the potential outcomes
when testing the sharp null hypothesis of no effect in 
Section \ref{sec:noeff}. 
We also introduced a more general testing
framework that relaxes the sharp null 
of constant non-zero effect to accommodate 
scale-free data, in Section \ref{sec:const}. 

We also discussed an additional class of estimands based on potential outcomes on a latent scale, in Section \ref{sec:laten}.
Estimands defined on the latent scale, generally suffer from non-identifiability issues due to the mapping from
the observed to the latent scales. As such, we caution practitioners when using them in applications.
Nonetheless, a latent variable formulation does allow us
to develop a Bayesian estimation framework for providing
posterior predictive estimates of the  induced  estimands on the observed scale.
We demonstrate this procedure using an ordered probit 
as well as a rank likelihood, in Section \ref{sec:bayes}, and illustrate the proposed methods with an application to educational outcomes in the General Social Survey, in Section \ref{dataexamp}.



